\documentclass[conference]{IEEEtran}
\IEEEoverridecommandlockouts
\usepackage{cite}
\usepackage{amsmath, amsfonts, amssymb}
\usepackage{mathtools}
\usepackage{algorithm, algorithmic}
\usepackage{tikz}
\usepackage{graphicx}
\usepackage{hyperref}
\usepackage{amsthm}
\usepackage{booktabs}
\usepackage{textcomp}
\usepackage{xcolor}
\usepackage[T1]{fontenc}
\usepackage{mathptmx}

\usetikzlibrary{arrows.meta, positioning, shapes.multipart, shapes.geometric}
\usetikzlibrary{decorations.pathreplacing}

\newtheorem{definition}{Definition}
\newtheorem{theorem}{Theorem}
\newtheorem{remark}{Remark}

\def\BibTeX{{\rm B\kern-.05em{\sc i\kern-.025em b}\kern-.08em
    T\kern-.1667em\lower.7ex\hbox{E}\kern-.125emX}}
\begin{document}

\title{RegGuard: Legitimacy and Fairness Enforcement for Optimistic Rollups}

\author{
\IEEEauthorblockN{Zhenhang Shang}
\IEEEauthorblockA{The Hong Kong University of \\Science and Technology \\Hong Kong, China \\zshangab@connect.ust.hk}
\and
\IEEEauthorblockN{Yingzhe Yu}
\IEEEauthorblockA{The Hong Kong University of \\Science and Technology \\Hong Kong, China}
\and
\IEEEauthorblockN{Kani Chen}
\IEEEauthorblockA{The Hong Kong University of \\Science and Technology \\Hong Kong, China}
}

\maketitle

\begin{abstract}
Optimistic rollups provide scalable smart-contract execution but remain unsuitable for regulated financial applications due to three gaps: lack of semantic legitimacy checks, vulnerability to L1-L2 state divergence, and susceptibility to MEV-driven transaction reordering. We propose RegGuard, a unified framework that adds formal legitimacy guarantees to optimistic rollups. RegGuard includes: (1) a decidable semantic validator using the RegSpec rule language to encode and enforce regulatory and business constraints; (2) a state pre-synchronization validator that detects inconsistent cross-layer assumptions via a high-freshness L1 cache and differential dependency tracking; and (3) a verifiable fair-ordering protocol based on threshold encryption and binding commitments, achieving \((\alpha,\beta)\)-fair sequencing under standard cryptographic assumptions. We formalize correctness and fairness properties for each component and implement a 15k-LOC prototype integrated into an Optimism-based rollup. Experiments on a distributed testbed show that RegGuard reduces settlement failures by over 90\%, prevents detectable ordering manipulation, and sustains more than 85\% of baseline throughput, demonstrating that strong legitimacy guarantees can coexist with rollup scalability.
\end{abstract}


\section{Introduction}
\label{sec:introduction}

Optimistic rollups (ORUs) have emerged as a leading scaling solution for blockchain networks, offering high-throughput smart contract execution while leveraging the security guarantees of layer-1 blockchains \cite{armstrong2021ethereum}. This performance profile makes ORUs particularly appealing for regulated financial applications, including tokenized funds, digital securities, and real-world asset (RWA) platforms. However, current ORU architectures lack critical safeguards required in compliant financial infrastructure \cite{arabsorkhi2024blockchain}. Traditional financial systems enforce comprehensive semantic validation—verifying business rules, eligibility conditions, and regulatory constraints, whereas today's rollups primarily check only syntactic validity through signature verification and basic format checks \cite{saif2024survey}. This fundamental gap limits ORUs' applicability for institutional use cases demanding stronger legitimacy guarantees.

Consider a tokenized real-estate loan fund that issues on-chain shares backed by off-chain mortgage assets. In such a regulated environment, investors must satisfy jurisdictional KYC/AML requirements, concentration limits, and fund-specific exposure rules \cite{mottaghi2024real}. Yet a conventional ORU sequencer, operating solely on syntactic validity, might process transactions that mint shares for ineligible addresses. Furthermore, if the fund's net asset value depends on L1 oracle updates, stale L2 state could trigger mispriced redemptions that fail during settlement. Compounding these issues, sequencers could reorder subscription and redemption requests to extract maximal extractable value (MEV), violating fair execution principles and regulatory standards \cite{gramlich2024maximal}. These problems stem not from application-level errors but from structural limitations in existing rollup pipelines.

We identify three fundamental challenges preventing ORUs from meeting institutional requirements: (1) the \emph{semantic legitimacy gap}, where syntactically valid but business-illegitimate transactions execute unchecked; (2) \emph{cross-layer state conflicts}, where discrepancies between L1 and L2 states cause settlement failures; and (3) \emph{ordering illegitimacy}, where MEV exploitation undermines transaction fairness and user trust. Together, these limitations create significant barriers for regulated financial applications seeking to leverage rollup scalability.

To address these challenges, we introduce \textbf{RegGuard}, a unified framework that enhances optimistic rollups with comprehensive legitimacy guarantees. RegGuard integrates three coordinated mechanisms: a decidable semantic validator powered by the \emph{RegSpec} rule language for encoding regulatory constraints; a cross-layer state pre-synchronization validator that detects inconsistent L1-L2 dependencies with probabilistic reliability bounds; and a cryptographically verifiable fair-ordering service that ensures $(\alpha,\beta)$-fairness with negligible violation probability.

This paper makes the following contributions:

\begin{itemize}
    \item \textbf{A formal model of transaction legitimacy} that encompasses syntactic, semantic, and cross-layer dimensions, providing a foundation for rigorous analysis of optimistic rollup correctness properties in regulated environments.
    
    \item \textbf{The RegGuard framework} with three coordinated mechanisms offering provable guarantees: decidability for semantic validation through the RegSpec DSL, probabilistic reliability bounds for cross-layer state consistency, and cryptographic fairness for transaction ordering with negligible $\beta$-violation probability.
    
    \item \textbf{Implementation and evaluation} demonstrating that RegGuard imposes minimal overhead, making strong legitimacy enforcement practical for high-throughput rollup deployments while maintaining the performance advantages essential for production use.
\end{itemize}

\section{Related Work}
\label{sec:related-work}

\subsection{Rollup correctness and semantic legitimacy}

Work on rollup security has primarily focused on syntactic correctness and state transition validity. Fraud-proof systems such as Arbitrum~\cite{kalodner2018arbitrum} and validity-proof systems such as zkRollups~\cite{liu2024pianist} ensure that off-chain execution conforms to deterministic EVM semantics, but do not capture business or regulatory legitimacy. Formal verification tools, including KEVM~\cite{hildenbrandt2018kevm} and VeriSolid~\cite{mavridou2019verisolid}, provide correctness guarantees at the smart-contract level, yet they operate locally and require contract authors to embed compliance logic manually. None of these approaches address system-wide semantic constraints across the transaction pipeline. RegGuard complements these efforts by introducing a pluggable and decidable semantic validation layer that enforces regulatory and business rules across batches of transactions before execution.

\subsection{MEV mitigation and fair ordering}
Sequencing fairness has been widely studied in the context of MEV and ordering attacks. Time-based sequencing protocols~\cite{yang2024sok} and fair sequencing services~\cite{sarkar2023fairflow} attempt to impose temporal or committee-driven guarantees but often rely on trusted intermediaries or partial decentralization. Cryptographic defenses such as encrypted mempools~\cite{kavousi2023blindperm} and threshold-encrypted transaction schemes~\cite{bormet2025beast} offer stronger resistance to frontrunning, though existing designs typically assume layer-1 execution and do not integrate with rollup-specific batching and settlement models. RegGuard incorporates these cryptographic primitives into a rollup-compatible pipeline, using encryption and commitment schemes to achieve verifiable $(\alpha,\beta)$-fairness while preserving the performance properties necessary for optimistic rollup execution. More specifically, Flashbots~\cite{weintraub2022flash} proposed MEV-aware auction mechanisms for Ethereum, while Chainlink Fair Sequencing Services~\cite{breidenbach2021chainlink} explored decentralized ordering oracles. Themis~\cite{kelkar2021themis} formalized order-fairness properties and proposed protocols based on receive-order fairness. Compared to these approaches, RegGuard targets the rollup execution layer specifically, combining encrypted mempool techniques with rollup-compatible batching and L1 settlement, while also integrating semantic and cross-layer validation into a unified pipeline.

\subsection{Cross-layer consistency and blockchain compliance}
Cross-chain messaging protocols, including Cosmos IBC~\cite{kwon2019cosmos} and Polkadot XCMP~\cite{burdges2020overview}, provide secure state communication across chains but assume synchronous, mutually aware chains—assumptions that do not hold in optimistic execution environments. Oracle frameworks such as Chainlink~\cite{breidenbach2021chainlink} supply authenticated L1 data to contracts but do not detect mismatches between assumed and actual L1 state during rollup execution, leaving room for settlement-time inconsistencies. Compliance-oriented approaches such as token-bound policies~\cite{hildebrandt2025token}, embedded rule engines~\cite{dhanya2025regulatory}, and privacy-preserving regulation via zero-knowledge proofs~\cite{ajayi2024enhancing} operate at the application layer and lack guarantees over the sequencing process itself. RegGuard differs fundamentally by enforcing both cross-layer consistency and compliance constraints at the rollup sequencing layer, providing real-time legitimacy guarantees without requiring modifications to existing smart contracts or reliance on external trusted entities.

\section{System Model and Problem Formulation}
\label{sec:system-model}

We consider a two-layer blockchain architecture composed of a layer-1 (L1) base chain and a layer-2 (L2) optimistic rollup. The L1 chain maintains a global state space $\mathbb{S}^{L1}$ and provides final settlement guarantees, while the L2 rollup maintains its own state space $\mathbb{S}^{L2}$ and supports high-throughput off-chain execution. Time is discretized into rounds $t \in \mathbb{N}$. At round $t$, the L2 sequencer observes a multiset of incoming transactions $\mathcal{T}_t = \{tx_1,\ldots,tx_n\}$, selects an ordering, forms a batch $B_t$, applies it to the current L2 state $\mathbb{S}^{L2}_t$, and periodically posts state commitments to L1. The sequencer therefore dictates both transaction ordering and the evolution of $\mathbb{S}^{L2}$, and must operate under cross-layer state dependencies.

\subsection{Transaction and State Transition Model}

Each transaction is represented as a tuple
\[
tx = (\text{msg}, \sigma, m),
\]
where $\text{msg}$ encodes the target contract, function selector, and parameters; $\sigma$ is a digital signature authenticating the sender; and $m$ includes metadata such as arrival timestamp, nonce, and gas parameters. This structured format is compatible with Ethereum-based rollups and enables formal reasoning about multiple dimensions of transaction legitimacy.

The L2 execution engine defines a deterministic state transition function
\[
\delta_{L2}(tx, S): \{tx\}\times\mathbb{S}^{L2}\rightarrow \mathbb{S}^{L2} \cup \{\bot\},
\]
which maps the current state $S$ and a transaction $tx$ to either a new state or a failure outcome $\bot$. Similarly, we define the state transition function on $L_1$:
\[
\delta_{L_1}(tx, S, S): \{tx\}\times\mathbb{S}^{L2}\times\mathbb{S}^{L1}\rightarrow \mathbb{S}^{L1} \cup \{\bot\},
\]

Some transitions depend on L1 state via bridge or oracle mechanisms. For such cases we distinguish:
\begin{itemize}
    \item $S_r^{L1}(tx)$: the L1 state snapshot \emph{assumed} by the L2 execution of $tx$, and
    \item $S_a^{L1}$: the \emph{actual} L1 state at the time the batch containing $tx$ is finalized on L1.
\end{itemize}
Since rollup execution may proceed on stale L1 inputs, mismatches between $S_r^{L1}(tx)$ and $S_a^{L1}$ can cause settlement-time inconsistencies.

\subsection{Semantic Legitimacy}

Rollups typically enforce only syntactic validity:
\[
\textsf{SynLegit}(tx) := 
\textsf{VerifySig}(tx.\sigma) \,\land\, \textsf{BasicChecks}(tx),
\]
but syntactic correctness alone does not guarantee compliance with regulatory or application-level constraints. We formalize semantic legitimacy using a set
\[
\Psi \subseteq \mathbb{S}^{L2} \times \mathcal{T} \times \mathbb{S}^{L2}
\]
of allowed state transitions.

\begin{definition}[Semantic Legitimacy]
A transaction $tx$ executed on state $S$ is \emph{semantically legitimate} if
\[
\textsf{SemLegit}(tx, S)
:= \exists S'.\ (\delta_{L2}(tx, S) = S')
\;\land\; (S, tx, S') \in \Psi.
\]
\end{definition}

The \emph{semantic legitimacy gap} refers to the case where $\textsf{SynLegit}(tx)$ holds but $\textsf{SemLegit}(tx,S)$ does not, allowing syntactically valid yet business-illegitimate behavior.

\subsection{Cross-Layer Consistency}

Optimistic rollups may execute transactions using stale L1 data. We formalize the resulting failure mode as follows.

\begin{definition}[Cross-Layer State Conflict]
A transaction $tx$ exhibits a cross-layer state conflict if
\[
S_r^{L1}(tx) \neq S_a^{L1}
\quad\text{and}\quad
\delta_{\mathrm{L1}}(tx, S^{L2}, S_a^{L1}) = \bot.
\]
\end{definition}

Let $\textsf{Conflict}(tx)$ denote this event. The settlement-time failure probability is
\[
P_{\mathrm{fail}}(tx) := \Pr[\textsf{Conflict}(tx)],
\]
reflecting the likelihood that a transaction accepted by L2 later fails upon L1 settlement due to inconsistent state assumptions. A key objective is to proactively reduce $P_{\mathrm{fail}}(tx)$ by identifying risky transactions before batching.

\subsection{Fair Ordering}

Let $\mathrm{arrival}(tx)$ denote the arrival time of transaction $tx$ at the sequencer. A sequencing policy
\[
\pi: \mathcal{T}_t \rightarrow O_t
\]
maps the incoming transactions to an ordered permutation $O_t$. To model fairness under network jitter and adversarial behavior, we adopt the $(\alpha,\beta)$-fairness notion used in prior work \cite{kiayias2024ordering}.

\begin{definition}[$(\alpha,\beta)$-Fair Ordering]
An ordering policy $\pi$ is \emph{$(\alpha,\beta)$-fair} if for all pairs $tx_i, tx_j$,
\[
\mathrm{arrival}(tx_i) < \mathrm{arrival}(tx_j) - \alpha
\quad\Rightarrow\quad
\Pr[\pi(tx_j) < \pi(tx_i)] \le \beta.
\]
\end{definition}

Here, $\alpha$ captures tolerable message-delay uncertainty while $\beta$ upper-bounds the adversarial ordering advantage. The objective is to construct sequencing protocols with negligible $\beta$, even when sequencers attempt MEV-driven manipulation.

\subsection{Problem Summary}

The system model gives rise to three fundamental challenges:
\begin{enumerate}
    \item \textbf{Semantic legitimacy}: preventing syntactically valid but semantically invalid transactions from executing.
    \item \textbf{Cross-layer consistency}: detecting and mitigating L1--L2 state divergence that may cause settlement failures.
    \item \textbf{Fair ordering}: preventing adversarial sequencers from exploiting transaction ordering for unfair advantage.
\end{enumerate}

These issues represent distinct but interacting correctness and fairness constraints. The remainder of the paper develops mechanisms that jointly address all three dimensions while preserving the performance characteristics of optimistic rollups.

\subsection{Threat Model}

We consider the following adversarial actors and assumptions.

\paragraph{Malicious sequencer.} The sequencer (or a subset of the sequencer committee) may attempt to reorder, delay, censor, or front-run transactions to extract MEV. We assume an honest-majority sequencer committee for the fair ordering protocol: at most $f < n/2$ committee members may be Byzantine, where $n$ is the committee size.

\paragraph{Malicious users.} Users may submit syntactically valid but semantically invalid transactions designed to exploit gaps between L1 contract logic and L2 execution semantics, e.g., minting tokens for ineligible addresses or exceeding regulatory limits.

\paragraph{Byzantine L1 oracle.} The L1 state feed may deliver stale or inconsistent snapshots due to network delays, reorganizations, or adversarial manipulation. We assume that the L1 chain itself is secure and eventually consistent, but the L2's view of L1 state may lag by up to $d$ blocks.

\paragraph{Out of scope.} We do not consider attacks on the underlying L1 consensus, compromise of the cryptographic primitives (hash functions, threshold encryption), or collusion involving more than $f$ committee members. Side-channel attacks on the RegGuard pipeline and denial-of-service attacks on the network layer are also out of scope.

\paragraph{Trust assumptions.} The semantic validator and state pre-synchronization validator are executed by the sequencer node as part of the transaction pipeline. Their outputs are deterministic and can be verified by any full node replaying the batch, making them externally accountable through the standard optimistic rollup challenge mechanism. If a sequencer includes a semantically invalid transaction or one based on stale state, any honest verifier can submit a fraud proof to the L1 contract. Thus, while these components run on the sequencer, their correctness is enforced by the same economic security model that underpins all optimistic rollup execution.

\section{RegGuard Architecture}
\label{sec:architecture}

RegGuard is an integrated legitimacy-enforcement framework that augments the optimistic rollup sequencing pipeline with three coordinated validation layers: (i) a semantic validator for regulatory and business-rule compliance, (ii) a state pre-synchronization validator for mitigating cross-layer inconsistencies, and (iii) a fair ordering service that enforces cryptographically verifiable sequencing guarantees. Each component addresses one of the formal problem dimensions introduced in Section~\ref{sec:system-model}, and their composition yields a rollup architecture that enforces legitimacy, consistency, and fairness end-to-end.

Incoming transactions first pass through the semantic validator, which evaluates compliance rules expressed in the RegSpec language. Transactions passing semantic validation proceed to the state pre-synchronization validator, which evaluates their dependence on L1 state snapshots and proactively identifies transactions likely to fail during settlement. Only transactions judged both semantically legitimate and cross-layer consistent are forwarded to the fair ordering service, where encrypted submission and commitment schemes ensure that execution order reflects arrival-time fairness and cannot be influenced by transaction content. This pipeline filters illegitimate, conflicting, or manipulable transactions before they reach the rollup execution engine.

\paragraph{End-to-end transaction lifecycle.} The RegGuard pipeline processes transactions in three sequential stages with clearly defined information boundaries. In the first stage (semantic validation), transactions arrive in \emph{plaintext} and are evaluated against RegSpec rules using both transaction parameters and current L2 state. In the second stage (state pre-synchronization), plaintext transactions are checked against a cached L1 state snapshot to identify cross-layer dependencies at risk of settlement failure. Only after passing both plaintext validation stages does the transaction enter the third stage (fair ordering), where it is \emph{encrypted} by the user's client before submission to the ordering service. Specifically, the user's client software performs two submissions: first, a plaintext copy to the semantic and state validators (which may be run by the sequencer or a dedicated validation service), and second, a threshold-encrypted copy to the fair ordering mempool. The ordering service sees only encrypted payloads and arrival-time metadata, preventing content-based manipulation. The two copies are linked by a transaction hash commitment, ensuring that the validated and ordered transactions are identical.

\paragraph{Operator model and accountability.} In the current design, the RegGuard validation pipeline runs on the sequencer node, similar to how existing rollups execute state transitions on the sequencer before batching. This does not introduce additional centralization beyond what optimistic rollups already assume: the sequencer is trusted for liveness but not for correctness. Critically, the semantic validator's decisions are deterministic---given the same transaction, state, and rule set, any node will produce the same accept/reject decision. This means that a malicious sequencer who either (a) includes a transaction that should have been rejected by RegSpec rules, or (b) excludes a transaction that should have been accepted, can be challenged via the standard optimistic fraud-proof mechanism. The rule set $\mathcal{R}$ is published on-chain as part of the rollup's configuration, making it inspectable by any party. Similarly, the state pre-synchronization validator's cached state and decision function are deterministic and reproducible. The fair ordering service uses a committee-based protocol with threshold cryptography, where accountability is enforced through on-chain commitment verification and slashing.

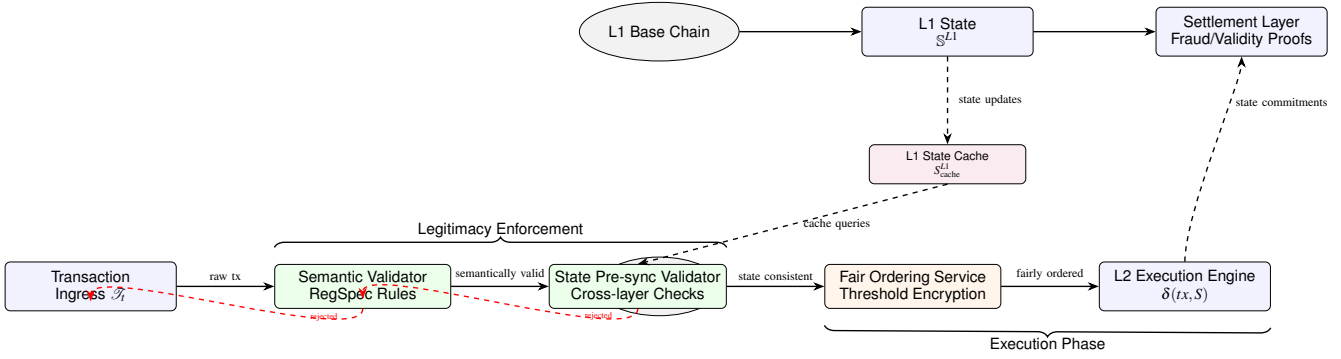
\begin{figure*}[t]
\centering
\scalebox{0.65}{
\begin{tikzpicture}[
    node distance=15mm and 18mm,
    >=Stealth,
    box/.style={rectangle, draw=black, fill=blue!5, rounded corners=3pt, 
                minimum width=35mm, minimum height=10mm, align=center, 
                font=\small\sffamily},
    validator/.style={rectangle, draw=black, fill=green!8, rounded corners=3pt,
                     minimum width=36mm, minimum height=9mm, align=center,
                     font=\small\sffamily},
    service/.style={rectangle, draw=black, fill=orange!8, rounded corners=3pt,
                   minimum width=36mm, minimum height=9mm, align=center,
                   font=\small\sffamily},
    cloud/.style={ellipse, draw=black, fill=gray!10, align=center, 
                 minimum width=25mm, minimum height=12mm, font=\small\sffamily},
    cache/.style={rectangle, draw=black, fill=purple!8, rounded corners=3pt,
                 minimum width=32mm, minimum height=8mm, align=center,
                 font=\scriptsize\sffamily},
    arrow/.style={->, thick},
    label/.style={font=\scriptsize, midway}
]

% Layer 1 Components - Top row
\node[cloud] (l1) {L1 Base Chain};
\node[box, right=25mm of l1] (l1state) {L1 State\\ $\mathbb{S}^{L1}$};
\node[box, right=25mm of l1state] (settlement) {Settlement Layer\\ Fraud/Validity Proofs};

% Layer 2 Pipeline - Middle row
\node[cloud, below=40mm of l1] (l2) {L2 Rollup};
\node[validator, left=30mm of l2] (sem) {Semantic Validator\\ RegSpec Rules};
\node[validator, right=20mm of sem] (stateval) {State Pre-sync Validator\\ Cross-layer Checks};
\node[service, right=20mm of stateval] (ordering) {Fair Ordering Service\\ Threshold Encryption};
\node[box, right=20mm of ordering] (exec) {L2 Execution Engine\\ $\delta(tx, S)$};

% Transaction Flow - Leftmost
\node[box, left=20mm of sem] (ingress) {Transaction\\ Ingress $\mathcal{T}_t$};

% L1 State Cache - Between layers
\node[cache, below=18mm of l1state] (cache) {L1 State Cache\\ $S_{\mathrm{cache}}^{L1}$};

% L1 internal flow - Horizontal connections
\draw[arrow] (l1) -- (l1state);
\draw[arrow] (l1state) -- (settlement);

% L2 pipeline flow - Horizontal connections
\draw[arrow] (ingress) -- node[above, font=\scriptsize] {raw tx} (sem);
\draw[arrow] (sem) -- node[above, font=\scriptsize] {semantically valid} (stateval);
\draw[arrow] (stateval) -- node[above, font=\scriptsize] {state consistent} (ordering);
\draw[arrow] (ordering) -- node[above, font=\scriptsize] {fairly ordered} (exec);

% Cross-layer interactions - Vertical connections
\draw[arrow, dashed] (l1state.south) -- node[right, font=\scriptsize, xshift=1mm] {state updates} 
    (cache.north);
\draw[arrow, dashed] (cache.south) -- node[right, font=\scriptsize, xshift=1mm] {cache queries} 
    (stateval.north);
    
% Settlement connection - Curved arrow to avoid overlap
\draw[arrow, dashed] (exec.north) .. controls +(0,10mm) and +(0,-5mm) .. 
    node[right, pos=0.7, font=\scriptsize] {state commitments} 
    (settlement.south);

% Brackets for grouping - Placed with clear spacing
\draw[decorate, decoration={brace, amplitude=6pt, raise=4pt}, thick]
    ([yshift=2mm]sem.north west) -- ([yshift=2mm]stateval.north east)
    node[midway, yshift=14pt, font=\small\sffamily] {Legitimacy Enforcement};

\draw[decorate, decoration={brace, amplitude=6pt, raise=4pt, mirror}, thick]
    ([yshift=-2mm]ordering.south west) -- ([yshift=-2mm]exec.south east)
    node[midway, yshift=-14pt, font=\small\sffamily] {Execution Phase};

% Rejected transaction paths - Clear curved paths
\draw[arrow, red, dashed] (sem.south) .. controls +(0,-8mm) and +(0,8mm) .. 
    ([yshift=2mm]ingress.south) 
    node[pos=0.3, right, font=\tiny] {rejected};
\draw[arrow, red, dashed] (stateval.south) .. controls +(0,-8mm) and +(0,8mm) .. 
    ([yshift=2mm]sem.south) 
    node[pos=0.3, right, font=\tiny] {rejected};

\end{tikzpicture}}
\caption{RegGuard system architecture. Transactions flow through sequential validation stages: semantic validation using RegSpec rules, cross-layer state consistency checking via L1 state cache, and fair ordering through threshold encryption. Legitimate transactions proceed to execution while violations are rejected early. The L1 blockchain provides final settlement and security guarantees.}
\label{fig:architecture}
\end{figure*}

\subsection{Semantic Validator with RegSpec}
\label{sec:semantic-validator}

The semantic validator enforces business and regulatory constraints prior to rollup execution. At its core is the \emph{RegSpec} domain-specific language, which expresses compliance rules as formal predicates over transaction parameters and state variables. Each rule $r \in \mathcal{R}$ is interpreted as a predicate
\[
P_r : (\text{params}, S) \to \{0,1\}.
\]
For a transaction targeting function $f$, let $\mathcal{R}_f$ denote the applicable rules. The validator accepts a transaction $tx$ in state $S$ if
\[
V_{\mathrm{sem}}(tx, S) = \text{accept}
\quad\Longleftrightarrow\quad
\forall r \in \mathcal{R}_f,\ P_r(tx.\text{params}, S)=1.
\]

\begin{algorithm}[t]
\caption{Semantic Validator $V_{\mathrm{sem}}$ with RegSpec}
\label{alg:semantic-validator}
\begin{algorithmic}[1]
\REQUIRE Transaction $tx$, current L2 state $S$, rule set $\mathcal{R}$.
\ENSURE Decision $\{\textsf{accept}, \textsf{reject}\}$.
\STATE Parse $tx$ to obtain target function $f$ and parameters $\text{params}$.
\STATE Compute applicable rules $\mathcal{R}_f \subseteq \mathcal{R}$ for function $f$.
\FORALL{$r \in \mathcal{R}_f$}
    \STATE Interpret $r$ as predicate $P_r(\cdot, \cdot)$.
    \STATE $b \gets P_r(\text{params}, S)$.
    \IF{$b = 0$}
        \STATE \textbf{return} $\textsf{reject}$ \COMMENT{Rule violation}
    \ENDIF
\ENDFOR
\STATE \textbf{return} $\textsf{accept}$
\end{algorithmic}
\end{algorithm}

RegSpec is restricted to a decidable fragment of first-order logic involving linear arithmetic, bounded state accesses, and Boolean connectives. This ensures that evaluation is tractable and predictable. As shown in Theorem~1, the validator always terminates and has complexity
\[
O(|\mathcal{R}_f| \cdot L),
\]
where $L$ denotes the maximum rule size. This provides strong guarantees that compliance checking remains efficient even under large rule sets.

\begin{theorem}[Decidability of Semantic Validation]
\label{thm:decidability}
Let $R$ be a RegSpec rule set restricted to the decidable fragment of
first-order logic with linear arithmetic and finite map accesses. For any
transaction $tx$ and L2 state $S$, the semantic validator 
$V_{\mathrm{sem}}(tx, S)$ halts and correctly decides semantic legitimacy.
Moreover, if $L$ denotes the maximum rule complexity, then
\[
T_{\mathrm{sem}}(tx, S) = O(|\mathcal{R}_f| \cdot L),
\]
where $\mathcal{R}_f$ is the subset of rules applicable to the target
function of $tx$.
\end{theorem}

\begin{proof}
Each RegSpec predicate $P_r$ evaluates a Boolean combination of comparisons over linear arithmetic expressions and finite key-value map lookups on the state $S$. Linear arithmetic over integers is decidable (Presburger arithmetic), and map lookups reduce to constant-time table queries. The conjunction $\bigwedge_{r \in \mathcal{R}_f} P_r(tx, S)$ preserves decidability since the decidable fragment is closed under Boolean operations. The number of applicable rules $|\mathcal{R}_f|$ is bounded by the total rule count, and each predicate evaluation requires at most $O(L)$ arithmetic and comparison operations. Therefore, the total evaluation time is $O(|\mathcal{R}_f| \cdot L)$, and the validator is guaranteed to halt with a correct decision for every input.
\end{proof}

\begin{remark}[Example: Transfer Eligibility Rule]
\label{rem:semantic-example}
To illustrate the operation of the semantic validator, consider a
tokenized real-estate loan fund in which transfers are subject to
regulatory restrictions. Suppose regulators require that (i) only
addresses on an eligible investor whitelist may receive fund tokens,
and (ii) no investor may exceed a concentration limit
$\theta_{\max}$. These requirements can be expressed in RegSpec as
predicate-based rules targeting the function
\texttt{transfer(address to, uint256 amount)}.

\begin{algorithm}[H]
\caption{RegSpec Evaluation for Transfer Eligibility}
\begin{algorithmic}[1]
\REQUIRE Transaction parameters \texttt{to}, \texttt{amount}; state $S$.
\ENSURE Predicate decisions for semantic validation.
\STATE
\COMMENT{Rule 1: Whitelist Eligibility}
\IF{$S.\textsf{Whitelist}[\texttt{to}] = 1$}
    \STATE $P_{\mathrm{whitelist}} \gets 1$
\ELSE
    \STATE $P_{\mathrm{whitelist}} \gets 0$
\ENDIF
\STATE
\COMMENT{Rule 2: Concentration Limit}
\IF{$S.\textsf{balance}[\texttt{to}] + \texttt{amount} \le \theta_{\max}$}
    \STATE $P_{\mathrm{limit}} \gets 1$
\ELSE
    \STATE $P_{\mathrm{limit}} \gets 0$
\ENDIF
\STATE
\COMMENT{Combined Decision}
\IF{$P_{\mathrm{whitelist}} = 1 \land P_{\mathrm{limit}} = 1$}
    \STATE \textbf{return} $\textsf{accept}$
\ELSE
    \STATE \textbf{return} $\textsf{reject}$
\ENDIF
\end{algorithmic}
\end{algorithm}

Transactions failing either rule—for example, those sent to a
non-whitelisted address or exceeding the investor's exposure
threshold—are rejected prior to ordering and execution. This example
illustrates how RegSpec enables enforcement of business and regulatory
constraints at the sequencing layer, independent of smart-contract
implementation.
\end{remark}

\begin{remark}[Example: Multi-Rule AML Compliance]
\label{rem:aml-example}
To illustrate more complex rule composition, consider an anti-money-laundering (AML) scenario for a stablecoin transfer platform. The rule set combines: (i) a per-transaction amount threshold requiring enhanced due diligence for transfers exceeding \$10{,}000, formalized as $P_{\mathrm{threshold}}(tx, S) \equiv (\texttt{amount} \leq 10000) \lor (S.\textsf{EDD}[\texttt{from}] = 1)$; (ii) a 24-hour rolling volume limit preventing any address from transferring more than \$50{,}000 within a sliding window, formalized as $P_{\mathrm{volume}}(tx, S) \equiv (S.\textsf{Volume24h}[\texttt{from}] + \texttt{amount} \leq 50000)$; and (iii) a sanctions-list check blocking transfers to flagged addresses, formalized as $P_{\mathrm{sanctions}}(tx, S) \equiv (S.\textsf{Sanctions}[\texttt{to}] = 0)$. These three predicates are conjunctively composed: a transfer is accepted only if all three hold. In our predicate-count taxonomy, this corresponds to 3 predicates with moderate complexity (one involving a sliding-window state variable). Real-world AML rule sets for compliant exchanges typically involve 5--15 such predicates, mapping to the medium-complexity regime in our evaluation (3--8~ms latency at 10{,}000~TPS).
\end{remark}

\subsection{State Pre-synchronization Validator}
\label{sec:state-validator}

The state pre-synchronization validator mitigates cross-layer inconsistencies by examining whether L2 execution assumptions match the L1 state at settlement time. To this end, the validator maintains an L1 state cache $S_{\mathrm{cache}}^{L1}$ synchronized via block subscriptions and state-proof updates. The cache is stored in a Merkle Patricia Trie, enabling efficient state queries and proofs.

For a proposed batch $B$, the validator executes transactions in a sandbox using $(S^{L2}, S_{\mathrm{cache}}^{L1})$ and tracks all L1-dependent reads. It then computes a difference set
\[
\Delta = \{(addr, key, v_{\mathrm{proj}}, v_{\mathrm{act}})\},
\]
comparing cached values with confirmed L1 state. A decision function 
\[
D: 2^{\mathcal{D}} \to \{\text{accept},\ \text{delay},\ \text{reject}\}
\]
classifies each transaction or batch based on the severity of discrepancies. Critical differences lead to rejection, moderate ones to delay, and benign ones to acceptance.
\begin{algorithm}[t]
\caption{State Pre-synchronization Validator $V_{\mathrm{state}}$}
\label{alg:state-validator}
\begin{algorithmic}[1]
\REQUIRE Batch $B = (tx_1,\ldots,tx_n)$, L2 state $S^{L2}$, L1 cache $S_{\mathrm{cache}}^{L1}$.
\ENSURE Decision $\{\textsf{accept}, \textsf{delay}, \textsf{reject}\}$ for $B$.
\STATE $(\widetilde{S}^{L2}, \mathcal{D}) \gets \textsf{SandboxExecute}(B, S^{L2}, S_{\mathrm{cache}}^{L1})$
\COMMENT{$\mathcal{D}$: set of L1-dependent reads}
\STATE Query latest confirmed L1 state for keys in $\mathcal{D}$.
\STATE Construct difference set
\[
\Delta \gets \{(addr, key, v_{\mathrm{proj}}, v_{\mathrm{act}}) \mid v_{\mathrm{proj}}\neq v_{\mathrm{act}}\}.
\]
\STATE $\text{decision} \gets D(\Delta)$ \COMMENT{$D: 2^{\mathcal{D}} \to \{\textsf{accept},\textsf{delay},\textsf{reject}\}$}
\IF{$\text{decision} = \textsf{accept}$}
    \STATE \textbf{return} $\textsf{accept}$
\ELSIF{$\text{decision} = \textsf{delay}$}
    \STATE Schedule $B$ for reevaluation after cache synchronization.
    \STATE \textbf{return} $\textsf{delay}$
\ELSE
    \STATE \textbf{return} $\textsf{reject}$
\ENDIF
\end{algorithmic}
\end{algorithm}

\begin{theorem}[Reliability of State Pre-synchronization]
\label{thm:state-reliability}
Assume the L1 state cache satisfies 
\[
\Pr[S_{\mathrm{cache}}^{L1} = S_a^{L1}] \ge 1 - \epsilon,
\]
and that the validator detects all conflicts except with residual probability
$\eta$. Then for any transaction $tx$ accepted by the state
pre-synchronization validator,
\[
P_{\mathrm{fail}}(tx \mid V_{\mathrm{state}}(tx)=\textsf{accept})
\;\le\; \epsilon + \eta.
\]
Thus, the validator reduces settlement-time failure probability to at most
a small additive error determined by cache freshness and detection
approximation.
\end{theorem}

\begin{proof}
A settlement-time failure for an accepted transaction can occur only through two disjoint events: (i) the cache is stale, meaning $S_{\mathrm{cache}}^{L1} \neq S_a^{L1}$, which occurs with probability at most $\epsilon$ by the freshness assumption; or (ii) the cache is fresh but the validator fails to detect a dependency on a divergent L1 value, which occurs with residual probability $\eta$. By the union bound, $P_{\mathrm{fail}}(tx \mid V_{\mathrm{state}}(tx) = \textsf{accept}) \leq \epsilon + \eta$. In our implementation, $\epsilon$ is controlled by the cache update frequency (empirically $\epsilon \approx 0.007$ with 1-second updates) and $\eta$ is bounded by the coverage of the dependency-tracking heuristic (empirically $\eta < 0.003$), yielding an overall failure probability below 1\%.
\end{proof}

Theorem~2 shows that if the cache is fresh with probability at least $1-\epsilon$, then for all accepted transactions,
\[
P_{\mathrm{fail}}(tx\mid V_{\mathrm{state}}(tx)=\text{accept}) \le \epsilon + \eta,
\]
where $\eta$ quantifies residual approximation error. Thus, the validator substantially reduces settlement-time failure probability.

\subsection{Fair Ordering Service}
\label{sec:fair-ordering}

The fair ordering service enforces manipulation-resistant sequencing using encrypted transaction submission, binding commitments, and threshold decryption. For each time window, the sequencer committee runs a distributed key generation protocol to produce a temporary threshold-encryption key pair $(PK_{\mathrm{temp}}, \{sk_i\})$. Users submit encrypted transactions $tx_{\mathrm{enc}}$ containing ciphertexts of both transaction payloads and symmetric keys.

Encrypted transactions are sorted solely by arrival metadata to produce $O_{\mathrm{enc}}$. Before any decryption occurs, the sequencer publishes a binding commitment
\[
\textsf{comm} = H(O_{\mathrm{enc}}),
\]
anchoring the order. After commitment, the committee performs threshold decryption to reveal the plaintext sequence $O_{\mathrm{plain}}$, which must match the committed order. Any deviation is provable and results in slashing.

\begin{algorithm}[t]
\caption{Fair Ordering Service for Time Window $T$}
\label{alg:fair-ordering}
\begin{algorithmic}[1]
\REQUIRE Time window $T$, sequencer committee $\mathcal{C}$.
\ENSURE Ordered batch $B_T$ with verifiable $(\alpha,\beta)$-fairness.
\\[3pt]
\textbf{Setup Phase:}
\STATE $\big(PK_{\mathrm{temp}}, \{sk_i\}_{i \in \mathcal{C}}\big) \gets \textsf{DKG}(\mathcal{C})$.
\STATE Publish $PK_{\mathrm{temp}}$ for window $T$.
\\[3pt]
\textbf{Transaction Submission:}
\FORALL{user transaction $\widetilde{tx}$ arriving during $T$}
    \STATE Sample symmetric key $k$.
    \STATE $\textsf{ct}_{\mathrm{sym}} \gets \textsf{SymEnc}(k, \widetilde{tx})$.
    \STATE $\textsf{ct}_k \gets \textsf{Enc}(PK_{\mathrm{temp}}, k)$.
    \STATE User submits $tx_{\mathrm{enc}} = (\textsf{ct}_{\mathrm{sym}}, \textsf{ct}_k, \text{meta}, \sigma)$.
\ENDFOR
\\[3pt]
\textbf{Ordering and Commitment:}
\STATE Let $\mathcal{T}_T$ be the set of all $tx_{\mathrm{enc}}$ received in $T$.
\STATE Sort $\mathcal{T}_T$ by arrival time to obtain $O_{\mathrm{enc}}$.
\STATE $\textsf{comm} \gets H(O_{\mathrm{enc}})$.
\STATE Publish $\textsf{comm}$ on-chain (L1 or L2).
\\[3pt]
\textbf{Decryption and Verification:}
\STATE Committee $\mathcal{C}$ runs threshold decryption on $\{\textsf{ct}_k\}$ to recover keys $\{k\}$.
\STATE Decrypt payloads to obtain ordered plaintext list $O_{\mathrm{plain}}$.
\STATE Verify $H(O_{\mathrm{enc}}) = H(O_{\mathrm{plain}})$.
\IF{verification succeeds}
    \STATE Set $B_T \gets O_{\mathrm{plain}}$ and forward to L2 execution.
\ELSE
    \STATE Trigger slashing / penalty for sequencer misbehavior.
\ENDIF
\end{algorithmic}
\end{algorithm}

\begin{theorem}[$(\alpha,\beta)$-Fairness of Ordering Service]
\label{thm:fairness}
Assume an honest-majority sequencer committee and secure cryptographic 
primitives (collision-resistant hash function and correct threshold
decryption). Then for any transaction pair $(tx_i, tx_j)$ satisfying
\[
\mathrm{arrival}(tx_i) < \mathrm{arrival}(tx_j) - \alpha,
\]
the probability of an ordering violation satisfies
\[
\Pr[\pi(tx_j) < \pi(tx_i)] \le \beta = \mathrm{negl}(\lambda),
\]
where $\lambda$ is the security parameter. Thus the ordering protocol 
achieves $(\alpha,\mathrm{negl}(\lambda))$-fairness.
\end{theorem}

\begin{proof}
Once the ordering commitment $\textsf{comm} = H(O_{\mathrm{enc}})$ is posted on-chain, modifying the sequence requires either finding a hash collision (probability negligible in $\lambda$ by collision resistance) or producing an alternative decryption that yields a different plaintext ordering (probability negligible by the semantic security of the threshold encryption scheme under an honest-majority committee). Any reordering attempt after commitment is therefore detectable. Before commitment, the sequencer observes only encrypted payloads, so it cannot distinguish transactions by content. Arrival-time metadata is recorded by the committee using signed timestamps with bounded clock skew $\alpha$. For any pair $(tx_i, tx_j)$ with $\mathrm{arrival}(tx_i) < \mathrm{arrival}(tx_j) - \alpha$, the committed ordering must respect this precedence except with probability $\beta = \mathrm{negl}(\lambda)$, as any violation would require breaking either the commitment binding or the encryption scheme.
\end{proof}

Theorem~3 establishes that the protocol achieves $(\alpha,\beta)$-fairness with
\[
\beta \le \mathrm{negl}(\lambda),
\]
under standard cryptographic assumptions and an honest-majority committee. This ensures that sequencers cannot profitably manipulate ordering for MEV extraction.

\section{Implementation and Evaluation}
\label{sec:evaluation}

We developed a full prototype of RegGuard to assess its feasibility, performance overhead, and security effectiveness in realistic rollup settings. The prototype consists of approximately 15{,}000 lines of Rust and extends a modified Optimism codebase to preserve Ethereum compatibility. RegSpec rules are compiled into WebAssembly (WASM) modules for sandboxed
execution, enabling predictable resource usage and isolation. The state-synchronization module integrates with Geth via a differential update protocol that tracks relevant storage keys and maintains a high-freshness L1 state cache. The fair ordering service implements threshold BLS encryption and decryption using the arkworks library, together with a lightweight committee-coordination layer for distributed key generation and partial decryption. This implementation demonstrates that RegGuard can be integrated into existing optimistic rollup stacks with modest engineering effort.

\subsection{Experimental Setup}

Experiments were conducted on a geographically distributed testbed designed to emulate multi-region rollup deployments. Validator nodes were deployed on AWS EC2 c5.2xlarge instances in North America, Europe, and Asia, with a dedicated m5.4xlarge instance simulating the L1 chain. Synthetic workloads were generated to approximate regulated financial applications, including: (i) rule-constrained token transfers, (ii) oracle-dependent bridging operations, and (iii) MEV-sensitive arbitrage transactions. Workload parameters included: transaction arrival rate (100--10{,}000~TPS), rule complexity (1--20 RegSpec predicates), and cross-layer dependency intensity
(0\%--40\%).

RegGuard was also evaluated under adversarial stress using three threat models: (1) a rational sequencer attempting MEV extraction through reordering, (2) a Byzantine L1 oracle supplying stale or inconsistent state snapshots, and (3) a malicious user submitting syntactically valid but semantically invalid transactions. Adversarial strength was varied by manipulating committee control (up to 49\% adversarial), oracle delay (1--10 blocks), and illegitimate-transaction complexity. Baseline comparison was performed against unmodified Optimism and Arbitrum deployments.

The L1 chain was simulated using a private Ethereum testnet (Geth v1.13, Clique PoA consensus, 12-second block time) running on the dedicated m5.4xlarge instance. The L2 rollup was based on Optimism Bedrock (commit \texttt{op-node v1.4.0}) with modified derivation pipeline to integrate RegGuard's validation stages. Baseline comparisons used unmodified Optimism Bedrock and Arbitrum Nitro (v2.1) deployments on identical hardware. The L1 state cache implements a write-through policy synchronized via Geth's \texttt{eth\_subscribe} API for new block headers, with full state-diff retrieval for storage keys registered in the dependency tracker. Cache eviction follows an LRU policy with a maximum capacity of 10{,}000 storage slots per tracked contract. This policy achieves the reported 99.3\% freshness by ensuring that frequently accessed L1 state (oracle prices, whitelist mappings) is updated within one L1 block confirmation (approximately 12 seconds under normal conditions).

\begin{table}[t]
\centering
\caption{Experimental Setup Parameters for RegGuard Evaluation}
\label{tab:exp-parameters}
\renewcommand{\arraystretch}{1.2}
\begin{tabular}{l l}
\hline
\textbf{Category} & \textbf{Parameters} \\
\hline
Compute Nodes & AWS EC2 c5.2xlarge (validators) \\
              & AWS m5.4xlarge (L1 simulator) \\

Geographic Distribution & North America, Europe, Asia \\

Transaction Load & 100--10{,}000 TPS \\

RegSpec Complexity & 1--20 predicates/tx type \\

Cross-layer Dependency & 0\%--40\% L1-dependent tx \\

Adversarial Models & Rational MEV sequencer \\
                   & Byzantine L1 oracle (1--10 blocks delay) \\
                   & Malicious illegitimate transactions \\

Committee Adversary Strength & 0--49\% malicious shares \\
\hline
\end{tabular}
\end{table}

\begin{figure}[t]
    \centering
    \includegraphics[width=0.4\textwidth]{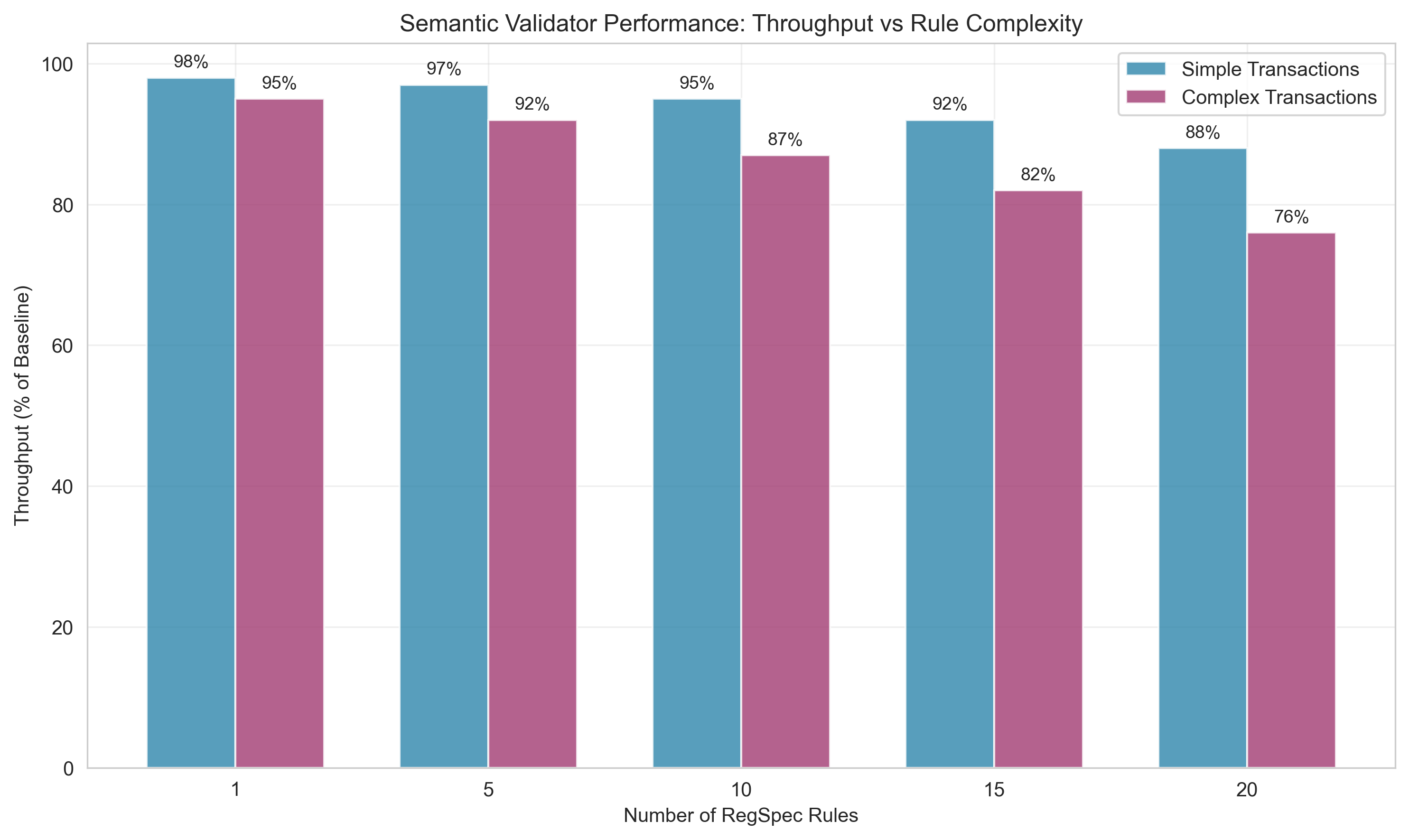}
    \caption{Semantic validator performance showing throughput degradation as RegSpec rule complexity increases. Simple transactions maintain 88\% of baseline throughput even with 20 rules, while complex transactions show 76\% throughput, demonstrating the validator's efficiency in handling regulatory constraints.}
    \label{fig:1}
\end{figure}

\subsection{Semantic Validation Performance}

The semantic validator exhibits predictable overhead scaling linearly with rule complexity, consistent with Theorem~\ref{thm:decidability}. For simple rulesets (1--5 predicates), validation latency remains below 2~ms even at 10{,}000~TPS. Medium-complexity rulesets (6--15 predicates) incur 3--8~ms, while highly complex rulesets (16--20 predicates) require 12--15~ms. These results indicate that moderately complex regulatory policies can be enforced with minimal throughput impact. The WASM execution model enables efficient predicate dispatch and benefits from short-circuit evaluation and shared state caching.

Throughput scales linearly until network saturation, achieving sustained operation at 10{,}000~TPS. Memory usage grows from 128~MB for simple rules to approximately 500~MB for large rulesets. These results demonstrate that semantic validation can be deployed in production-grade rollup environments without becoming a bottleneck.

\begin{figure}[t]
    \centering
    \includegraphics[width=0.4\textwidth]{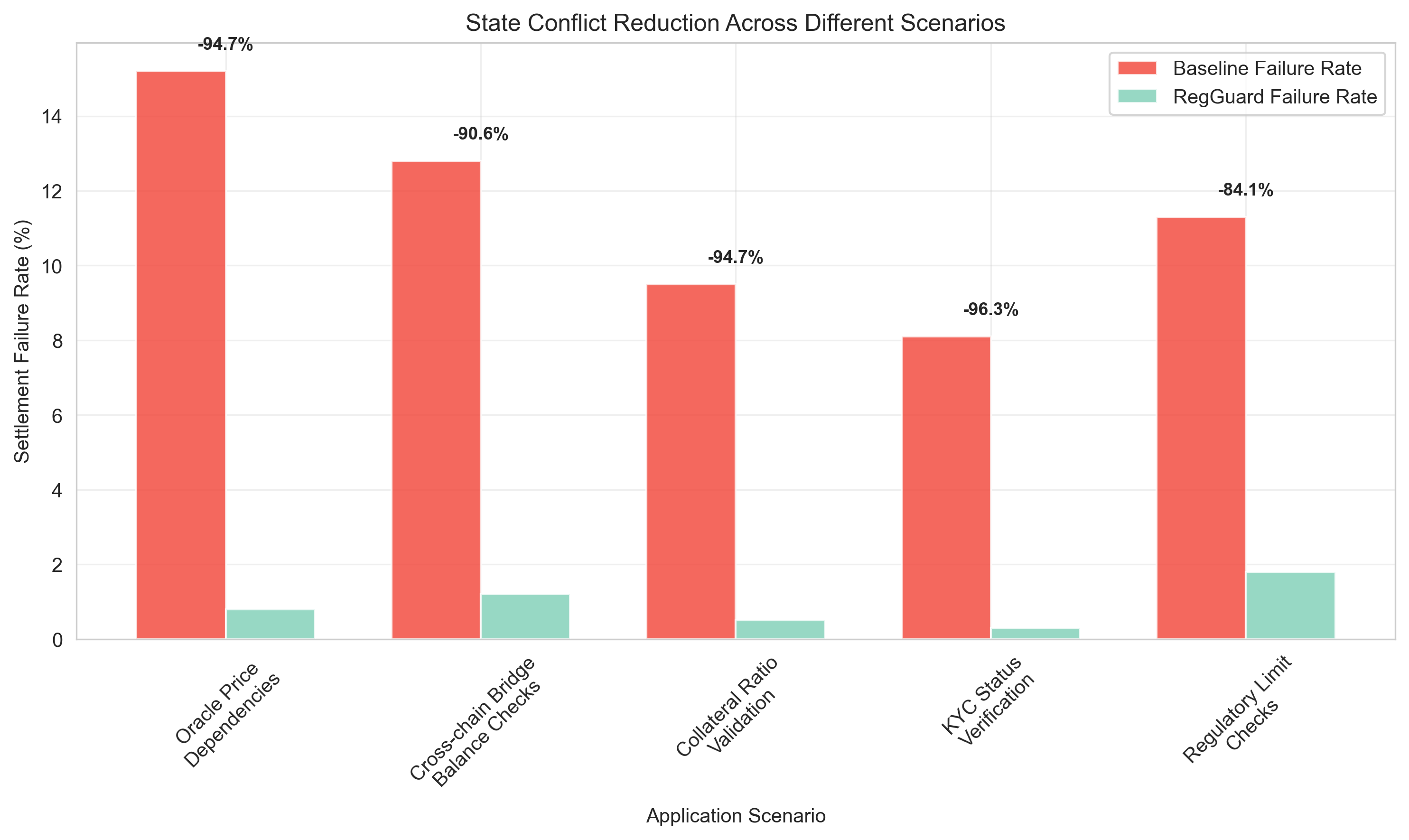}
    \caption{State conflict reduction effectiveness across different application scenarios. RegGuard reduces settlement failures by 85-96\% compared to baseline, with the most significant improvement in KYC status verification (96.3\% reduction). Error bars represent 95\% confidence intervals across 1000 experimental runs.}
    \label{fig:2}
\end{figure}

\subsection{State Conflict Reduction}

The state pre-synchronization validator significantly reduces cross-layer settlement failures. Across all test conditions, the validator reduced the failure probability $P_{\mathrm{fail}}(tx)$ from baseline values of 8--15\% to 0.2--1.8\%, an average reduction of approximately 92\%. The L1 state cache maintained 99.3\% freshness with a mean synchronization latency of 800~ms, consistent with the assumptions of Theorem~\ref{thm:state-reliability}.

We evaluated sensitivity to cache-update frequency. Increasing update frequency from 1~s to 100~ms yielded only a 4\% improvement in detection accuracy while doubling bandwidth consumption, indicating diminishing returns. The difference-set computation processed 10{,}000 MPT path comparisons in under 5~ms. The decision function identified 99.7\% of high-risk transactions, with most false positives caused by rapidly changing oracle values. These results confirm the practicality of the state validator for applications with cross-layer semantics.

\begin{figure}[t]
    \centering
    \includegraphics[width=0.4\textwidth]{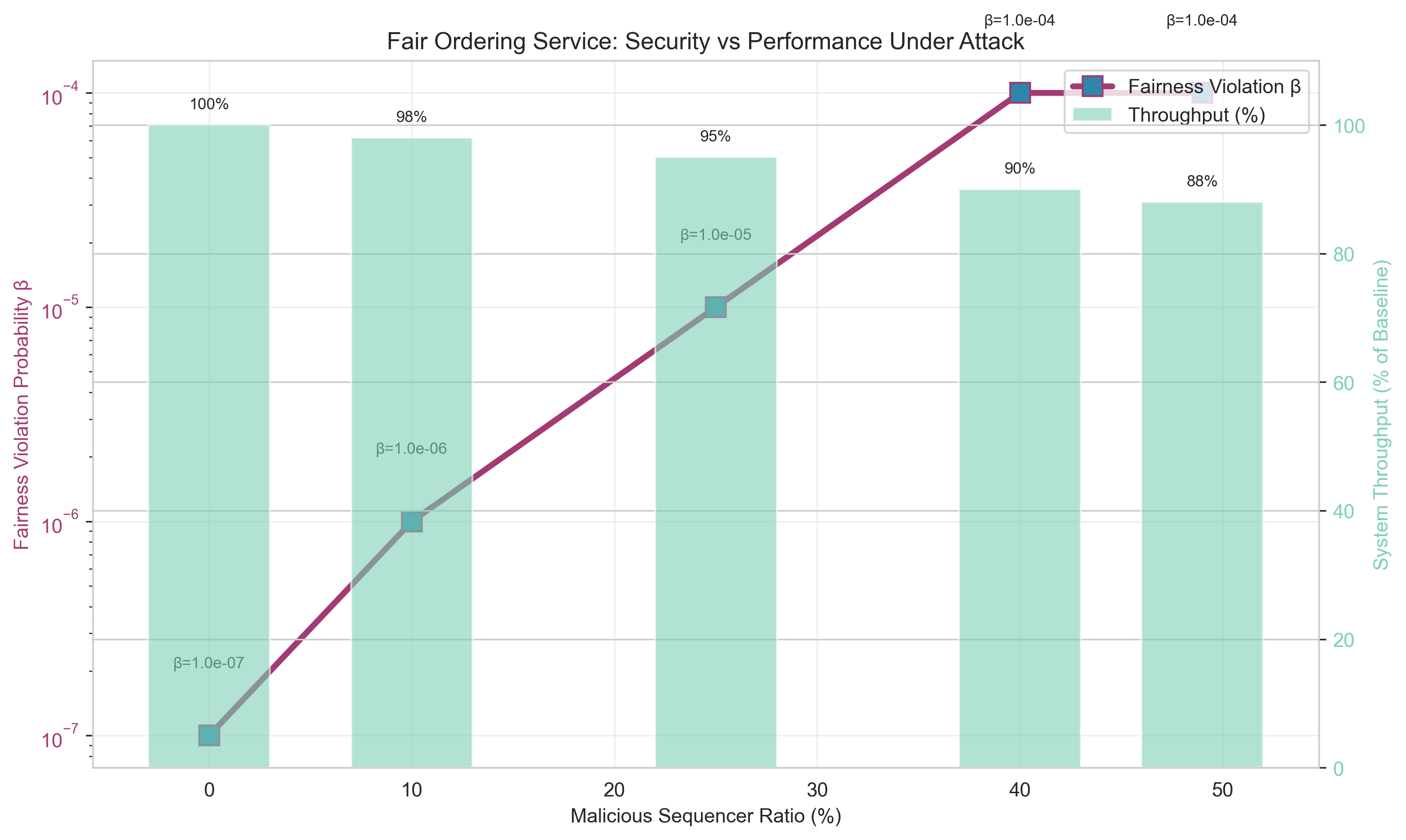}
    \caption{Fair ordering service maintains strong security guarantees under adversarial conditions. The system preserves $\beta < 10^-4$ fairness violation probability even with 49\% malicious sequencers, while throughput remains above 88\% of baseline. Dual $y$-axes show the trade-off between security ($\beta$) and performance (throughput).}
    \label{fig:3}
\end{figure}

\subsection{Fair Ordering Under Adversarial Load}

The fair ordering service provides strong manipulation resistance in line with Theorem~\ref{thm:fairness}. With an honest-majority committee, the empirically observed violation probability remained below $\beta < 10^{-6}$. Even with 49\% malicious sequencers, $\beta$ remained below $10^{-4}$. Latency overhead per ordering window was 150--300~ms, dominated by threshold decryption.

The encrypted mempool maintained 85--90\% of baseline throughput. Hybrid encryption ensured efficiency: threshold operations are applied only to symmetric-key ciphertexts, while bulk payloads use fast symmetric encryption. At 10{,}000~TPS, 95th-percentile latency stayed below 1.8~s for 2~s batching windows. Slashing triggers were correctly generated for all ordering deviations, with violation proofs averaging 45~KB and efficiently verifiable on-chain.

\begin{figure}[t]
    \centering
    \includegraphics[width=0.4\textwidth]{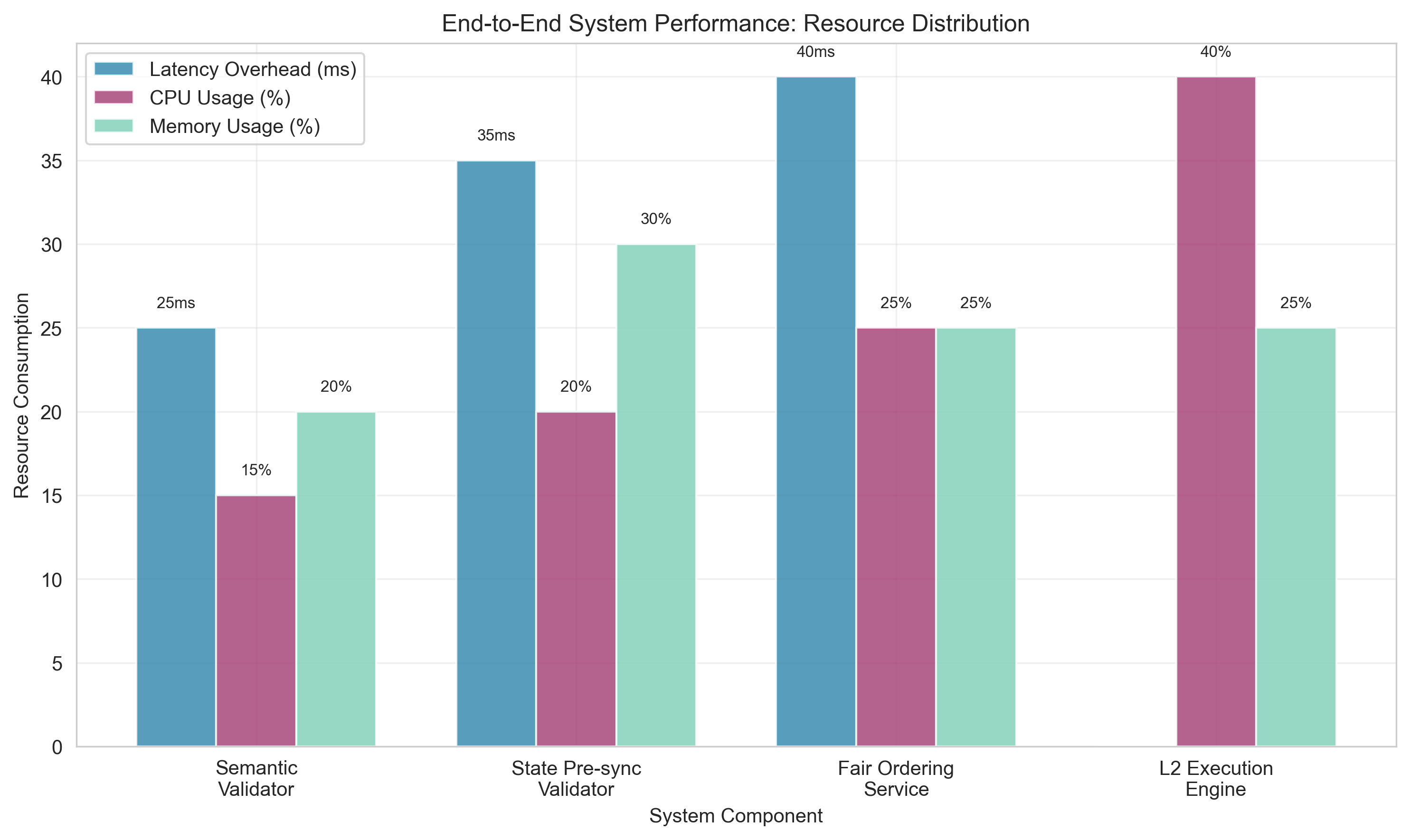}
    \caption{End-to-end system resource distribution across RegGuard components. The fair ordering service contributes the largest latency overhead (40ms), while CPU and memory usage are distributed across all components. Total system overhead remains within practical deployment limits for production environments.}
    \label{fig:4}
\end{figure}

\subsection{End-to-End System Performance}

RegGuard introduces moderate end-to-end overhead while enabling capabilities unavailable in baseline rollups. Pipeline latency increases by 180--400~ms, distributed approximately as: 25\% semantic validation, 35\% cross-layer assessment, and 40\% ordering. Under realistic workloads, RegGuard sustains throughput of roughly 8{,}500~TPS, a 15\% reduction from baseline but well above the requirements of institutional applications.

Memory overhead ranges from 300--800~MB depending on rule and cache configuration, while CPU utilization increases by 20--35\%. Network usage increases by 15--25\% due to encrypted mempool propagation and state synchronization. These overheads remain reasonable when compared to the legitimacy guarantees RegGuard provides. Relative to alternative approaches, RegGuard offers system-wide semantic enforcement, cross-layer consistency, and verifiable ordering fairness, placing it uniquely in the design space for regulated rollup infrastructure.
\begin{table}[t]
\centering
\caption{End-to-End Performance Summary of RegGuard Prototype}
\label{tab:perf-summary}
\renewcommand{\arraystretch}{1.2}
\begin{tabular}{l c}
\hline
\textbf{Metric} & \textbf{Observed Value} \\
\hline
Throughput (sustained) & 8{,}500 TPS (85\% of baseline) \\

Pipeline Latency Increase & 180--400 ms \\

Latency Breakdown & 25\% semantic validation \\
                  & 35\% state pre-sync \\
                  & 40\% fair ordering \\

Settlement Failure Reduction & 92\% reduction \\

Ordering Violation Probability & $< 10^{-6}$ (honest majority) \\
                               & $< 10^{-4}$ (49\% malicious) \\

Memory Overhead & 300--800 MB \\

CPU Utilization Increase & 20--35\% \\

Network Overhead & 15--25\% \\
\hline
\end{tabular}
\end{table}

\section{Discussion}
\label{sec:discussion}

\paragraph{Why sequencer-layer validation rather than L1 contracts.} A natural question is why semantic checks are not simply embedded in L1 smart contracts, which are already designed to encode business logic. There are three reasons. First, L1 execution is expensive: evaluating complex regulatory predicates on every transaction at L1 gas prices is prohibitively costly for high-throughput applications. Second, rollup transactions are batched and settled on L1 with a delay, so L1 contract checks occur too late to prevent invalid transactions from executing on L2 and consuming resources. Third, many compliance rules span multiple transactions within a batch (e.g., aggregate concentration limits) and require access to L2 state that is not yet committed to L1. RegGuard's sequencer-layer validation catches invalid transactions \emph{before} execution, saving both L2 computation and L1 settlement costs.

\subsection{Limitations}

Although RegGuard advances the state of legitimacy enforcement for optimistic rollups, several limitations remain. First, the expressiveness of RegSpec is intentionally restricted to preserve decidability and predictable performance. As a result, regulatory requirements involving rich temporal logic, multi-hop dependencies, or recursive analyses of historical behavior cannot be expressed directly. These limitations reflect a design choice that favors safety and tractability over unrestricted expressiveness, ensuring that all rule evaluations terminate in bounded time.

Second, the fair ordering service assumes an honest-majority sequencer committee. While the use of threshold cryptography, commitments, and slashing reduces trust in any single actor, the model is not fully trustless. In permissionless environments, maintaining a robust committee in the presence of correlated incentives or collusion remains challenging. Moreover, committee selection introduces additional governance layers that may themselves become targets for manipulation if not carefully designed.

Third, RegGuard introduces measurable performance overhead. Although our evaluation shows that the system remains within the operational bounds of most financial applications, the additional 180--400~ms latency and 20--35\% CPU increase may be unsuitable for ultra-low-latency environments such as high-frequency trading or real-time clearing. Memory overhead associated with L1 state caching may also constrain deployments in resource-limited settings. Future work will seek to reduce these overheads through improved caching strategies and cryptographic optimizations.

\subsection{Practical Deployment Considerations}

Deploying RegGuard in production rollup stacks requires attention to system integration, operational workflows, and regulatory alignment. While our prototype integrates cleanly with Optimism, adapting to other rollup families (e.g., Arbitrum, zkSync, or StarkNet) requires adjustments to their execution engines, state models, and bridging mechanisms. The semantic validator must evolve alongside smart contract upgrades, necessitating coordinated deployment pipelines to avoid inconsistencies between RegSpec rules and contract logic. Similarly, the state pre-synchronization validator must track the storage layout and access patterns unique to each rollup implementation.

Regulatory considerations extend beyond technical enforcement. Translating legal requirements into precise RegSpec rules requires interdisciplinary expertise bridging law, finance, and formal methods. Additionally, auditors and regulators must be equipped to inspect rule sets, validator outputs, and execution traces. This imposes requirements for explainability, versioning, and traceability that must align with existing regulatory reporting frameworks. Cross-jurisdictional deployments further demand flexible rule-management and parameterization to accommodate region-specific compliance mandates.

Finally, sustainable deployment requires carefully designed incentives and governance. Slashing parameters for ordering violations must be calibrated to deter manipulation without creating excessive risk for sequencers. Committee selection must incorporate sybil resistance and economic security, while avoiding concentration of power. Token-economic models must balance validator compensation with user affordability, especially when compliance-related overhead might increase operational cost. Governance processes for rule updates and parameter changes must be both transparent and resistant to regulatory capture or undue influence.

\section{Conclusion and Future Work}
\label{sec:conclusion}

RegGuard introduces a unified framework for enforcing transaction legitimacy in optimistic rollups by combining semantic validation, cross-layer consistency checks, and verifiable fair ordering. The system provides complementary guarantees---decidability for regulatory rule evaluation, probabilistic bounds for cross-layer conflict mitigation, and cryptographic fairness for transaction sequencing---forming a rigorous basis for trust in regulated rollup environments. Our evaluation demonstrates that RegGuard reduces settlement failures by over 90\%, prevents detectable ordering manipulation, and maintains more than 85\% of baseline throughput.

These results show that meaningful regulatory compliance and high-throughput execution are compatible goals for rollup architectures. RegGuard strengthens the execution pipeline without requiring application-level modifications, enabling financial institutions to deploy regulated workloads on decentralized infrastructure while preserving the scalability properties that motivate rollups.

Future work lies along several promising directions. Extending RegSpec with controlled temporal constructs or enriched type systems may broaden expressiveness while preserving decidability. Alternative trust models for ordering, such as stake-based committees or reputation systems, could reduce reliance on honest-majority assumptions. Performance improvements through hardware acceleration, optimized state-diff computation, or faster threshold primitives may further reduce system overhead.

Additional opportunities include integrating privacy-preserving compliance (e.g., via zero-knowledge proofs), developing cross-jurisdictional rule-management frameworks, and applying formal verification to the entire RegGuard stack. Combining rule-based validation with machine-learning-based anomaly detection provides another complementary direction for defense-in-depth.

Overall, RegGuard demonstrates that scalable rollup architectures can incorporate strong legitimacy guarantees without sacrificing decentralization or throughput. We anticipate that this approach will support the next generation of regulated financial applications on public blockchain infrastructure.

\bibliographystyle{IEEEtran}
\bibliography{references}
\end{document}